\documentclass[lettersize,journal]{IEEEtran}
\usepackage{amsmath,amsfonts}
\usepackage{algorithmic}
\usepackage{array}
\usepackage[caption=false,font=normalsize,labelfont=sf,textfont=sf]{subfig}
\usepackage{textcomp}
\usepackage{stfloats}
\usepackage{url}
\usepackage{verbatim}
\usepackage{graphicx}

\usepackage{braket}
\usepackage[nolist, printonlyused]{acronym}
\usepackage{amsthm}
\newcommand{\argmax}{\mathop{\mathrm{arg\,max}}}
\DeclareMathOperator{\ld}{ld}

\DeclareMathOperator{\tr}{Tr}
\usepackage{diagbox}   
\usepackage{multirow}  
\usepackage{amssymb}
\usepackage{xcolor}   

\def\BibTeX{{\rm B\kern-.05em{\sc i\kern-.025em b}\kern-.08em
    T\kern-.1667em\lower.7ex\hbox{E}\kern-.125emX}}
\usepackage{balance}
\newtheorem{proposition}{Proposition}

\begin{document}

\begin{acronym}[QKD] 


\acro{QKD}  {Quantum Key Distribution}
\acro{cqq}  {classical-quantum-quantum}
\acro{E91}  {Ekert 1991}
\acro{CHSH} {Clauser–Horne–Shimony–Holt}
\acro{LOCC} {local operations and classical communications}
\acro{RFC}  {Request for Comments}

\end{acronym}


\title{Processing Entangled Links Into Secure Cryptographic Keys}
\author{Marcel Kokorsch and Guido Dietl, \IEEEmembership{Senior Member, IEEE}
\thanks{
Marcel Kokorsch and Guido Dietl are with the Professorship of Satellite Communication and Radar Systems, University of Würzburg, 97074 Würzburg, Germany  (e-mail: marcel.kokorsch@uni-wuerzburg.de; guido.dietl@uni-wuerzburg.de)

This is the author’s accepted manuscript of a paper published in the IEEE Journal on Selected Areas in Communications (Volume: 44). The final version is available at https://doi.org/10.1109/JSAC.2026.3707822.

© 2026 IEEE. Personal use of this material is permitted. Permission
from IEEE must be obtained for all other uses, in any current or future
media, including reprinting/republishing this material for advertising or
promotional purposes, creating new collective works, for resale or
redistribution to servers or lists, or reuse of any copyrighted
component of this work in other works.

}
}

\maketitle

\begin{abstract}
The following paper presents a holistic approach to the processing of entangled links within entanglement based quantum key distribution (QKD) protocols, whose security relies on the Bell inequality. We investigate the collective impact of the whole processing chain on the final secure key rate. This includes the quantum mechanical preprocessing in the form of entanglement distillation, processing of the entangled states via measurements, and the necessary classical postprocessing based on the measurement results. Our investigations are based on the principal idea of the Ekert 1991 protocol and utilize the secret key capacity introduced by Devetak and Winter in 2005. Our results include a proof on what measurement bases need to be chosen to achieve this capacity for the case of Werner states. We also propose a new processing strategy for how the two users of the QKD system can choose their measurement bases to achieve higher key generation rates for certain scenarios. Furthermore, we answer the question on how many recurrence based entanglement distillation iterations are optimal. This is achieved by proposing a unified formalism and describing the whole processing chain that can be used to make quantitative statements on the relation between the quality and quantity of entangled but noisy quantum states used for the generation of secure keys.
\end{abstract}


\section{Introduction}
\IEEEPARstart{W}{ithin} the current second quantum revolution \cite{Dowling2003}, \ac{QKD} is one of the most promising applications of quantum communications. In this context, global communication networks are of significant interest. The current limitation of such \ac{QKD} systems is their limited distance \cite{Pirandola2017}, which makes entanglement based protocols, like, e.g., the \ac{E91} \cite{Ekert1991} protocol, especially interesting as they pose the possibility to integrate \ac{QKD} in potential future large scale entanglement-based quantum communication networks \cite{Kozlowski2023}. Furthermore, entanglement based protocols that draw their security on the violation of the Bell inequality \cite{Bell1964}, like the \ac{E91} protocol \cite{Ekert1991}, pose the great advantage of allowing security proofs whose security is device independent \cite{Pironio2009}. The main challenge in the implementation of these systems is the unstable nature of quantum systems and their interaction with the environment resulting in noise \cite{Cacciapuoti2020, Chehimi2022}. To solve this, one can find the idea of using decoherence-free subspaces \cite{Dorner2008} within the literature or, what seems to be more applicable, entanglement distillation \cite{Duer2007,Bennett1996,Deutsch1996} or quantum error correction \cite{Roffe2019}. According to a study on optimal architectures for long distance quantum communication \cite{Muralidharan2016} and the \ac{RFC} 9340 \cite{Kozlowski2023} on the architectural principles of a quantum internet (proposed by the Quantum Internet Research Group), the first, and partly the second, generation of quantum communication systems are supposed to use entanglement distillation to ensure error tolerance. This reason and the fact that entanglement distillation has already been demonstrated experimentally \cite{Pan2001}, highly motivates further research in the optimal application and implementation of these protocols.

The general idea of entanglement distillation is to take several noisy quantum states and to distill them into a fewer amount of less noisy quantum states, purely by \ac{LOCC}. Within the literature, there are many contributions investigating entanglement distillation, e.g., \cite{Duer2007,Rozpedek2018,Chen2024,Lo1999}. Some look into these protocols isolated \cite{Duer2007,Rozpedek2018}, while others investigate their integration into quantum communication systems with the specific objective to surpass a certain fidelity threshold, e.g., \cite{Chen2024}. The Lo-Chau protocol \cite{Lo1999} uses it also to prove the security of \ac{QKD} systems. However, with regard to the implementation of real-world, large-scale, entanglement-based \ac{QKD} schemes, whose overarching objective is the efficient generation of secure keys, there are still unresolved questions on how to optimally employ entanglement distillation schemes. 

To quantify the efficiency of the generation process of secure keys, information-theoretic investigations on the capacity of entangled links, such as the Devetak-Winter rate \cite{Devetak2005}, exist. A key challenge of these models is ensuring that the generated key remains secure against any potential attack by an eavesdropper. There are several works within the literature that deal with this by computing lower bounds on the achievable asymptotic key rate through the means of finding the worst case quantum state by semidefinite programming, based on the observed measurement distributions, e.g., \cite{Tan2021,Brown2021,Metger2023}. However, these approaches mainly focus on proving the security of the investigated \ac{QKD} protocols rather than investigating the efficient implementation of them. Furthermore, since they are based on the observed measurement distributions, they are not fully capable of capturing the influence of how the entangled states are processed into these distributions. 

An alternative technique to compute the secret key capacity directly from the density operator describing the distributed noisy quantum state has been proposed in \cite{Ettinger2025} and has already been used to investigate the influence of entanglement distillation beyond the increase of fidelity on the achievable secure key rate in \cite{Kokorsch2025}. However, several open questions remain regarding how entangled links should be processed via measurements in order to achieve the maximal secure key generation rate. In particular, the prior works \cite{Ettinger2025} and \cite{Kokorsch2025} did not account for the impact of sifting on the key rate. Also, it remains unclear which measurement strategies attain the capacity bounds of the Devetak–Winter rate across different classes of noisy states with respect to the chosen measurement bases. Moreover, the role of entanglement distillation in this context has also not been investigated in a general, protocol-independent framework that captures its interplay with these effects. Several of these questions are interlocked and can therefore not be answered individually.

Within this paper, we take a holistic approach to solve this problem. We will investigate the entire processing chain required to generate a secure key from a set of entangled, yet noisy, quantum states. This means, we consider the quantum mechanical preprocessing, i.e., entanglement distillation, the processing of the quantum states, i.e., their measurement with the goal of extracting correlated bit strings, and the necessary classical postprocessing to turn these correlated bit strings into usable secure keys. By doing so, we contribute towards a unified framework to identify the optimal implementation strategies for this class of \ac{QKD} protocols. 

This paper is structured as follows: Section \ref{CH:System_Model} introduces the system model and describes the rate with which a secure key can be generated. This involves the process of generating entangled links in Subsection \ref{CH:Entangled_Link_Generation}, the quantum mechanical preprocessing in the form of entanglement distillation in Subsection \ref{CH:r_ent}, the processing of these quantum states into classical correlations in Subsection \ref{CH:r_sift}, and the classical postprocessing in Subsection \ref{CH:r_key}. In Section \ref{CH:Optimal_key_rate}, we derive the optimal measurement bases for Werner states, required to attain the capacity of the Devetak-Winter rate. We use these results in Section \ref {CH:Processing} to investigate and compare different processing strategies. In Section \ref{CH:Optimal_k}, we describe a methodology on how to determine the optimal number of entanglement distillation iterations to be performed, with respect to maximizing the resulting secure key rate. We briefly discuss how the modification of Werner states into more general Bell diagonal states by entanglement distillation influences all our results in Section \ref{CH:Post_Distillation}, and in Section \ref{CH:Conclusion}, we finally conclude the paper.

Within this work, we will be using the following notation. For the three parties Alice (A), Bob (B), and Eve (E), we denote their corresponding subsystems with capital letters A, B, and E, respectively. The set of all density operators, representing quantum states within the Hilbert space $\cal{H}$, will be denoted as $\cal{D}(\cal{H})$. Composite Hilbert spaces will be abbreviated in the form $\mathcal{H}_A \otimes \mathcal{H}_B = \mathcal{H}_{AB}$. For the specific measurement results obtained by projection valued measurement using orthonormal bases, we use small Greek letters with superscripts, e.g., $\alpha$ and $\alpha^\perp$ for A. With small Latin letters we denote the unspecific measurement results which may correspond to either of the two potential outcomes, e.g., $a\in \{ \alpha,\alpha^\perp\}$. If two quantum states differ but have the same fidelity, we refer to them as having different shapes. The symbol $I_n$ denotes the $n$-dimensional identity operator and $\Pi_n = \frac{1}{n} I_n$ is the density operator of the $n$-dimensional maximally mixed state. For an easier notation, we write the complex conjugated inside the bras and kets, i.e., $\ket{\psi^*} = \ket{\psi}^*$. The eigenvalue spectrum of an operator $O$ is denoted by $\sigma[O]$. $\tr[O]$ denotes the trace of the operator $O$ and $\tr_X[O]$ is the partial trace with respect to the subsystem $X$. Lastly, we define $\mathrm{ld}(\square)$ to be the logarithm to the base two.

\section{System Model and Problem Formulation}\label{CH:System_Model}

Within this work, we want to investigate how an entangled link between two parties A and B can be processed into a secure key,  with a potential eavesdropper E attacking the system. To do so, we will not only investigate the processing of quantum states into classical information, but also the influence of quantum mechanical preprocessing as well as classical postprocessing. We base our analysis on the idea of the entanglement based \ac{QKD} protocol \ac{E91} \cite{Ekert1991}. Thus, we define the main key performance criteria of the system to be the achievable secure key generation rate given by
\begin{equation}\label{eq:r}
    r = \frac{\ell}{M},
\end{equation}  
where $\ell$ is the number of secure key bits that can be generated in relation to the number $M$ of successfully distributed quantum states between A and B. 

This process can further be divided into four steps:
\begin{enumerate}
    \item An entangled link between A and B is generated, resulting in $M \in \mathbb{N}$ initially distributed, potentially noisy, bipartite quantum states $\varrho_{AB}^0$ shared by A and B. 
    \item The total amount of $M$ initially distributed states $\varrho_{AB}^0$ can be quantum mechanically preprocessed by  means of $k \in \mathbb N_0$ iterations of entanglement distillation, resulting in $N \leq M$ quantum states $\varrho_{AB}^k$. 
    \item A and B process their $N$ quantum states $\varrho_{AB}^k$ by first performing projection valued measurements, resulting in $N$ classical measurement results. Next, they perform sifting, reducing the amount of measurements that can be used to generate the secure key to $n < N$. The remaining $N-n$ measurements are either used for proving the security of the protocol or are unavoidably wasted measurements.    
    \item Lastly, classical postprocessing is performed to distill a usable secure key, which reduces the $n$ measurement results to $\ell \leq n$ bits of a secure key. 
\end{enumerate}

The quantities denoting the amount of quantum states and classical bits $M, N, n,$ and $\ell$ will be natural numbers in any real system. However, in the following analysis we assume the asymptotic regime $M \to \infty$, as this allows us to work consistently with information-theoretic quantities defined in the same limit and avoid complications arising from finite-size effects such as rounding. This means that the process of the quantum mechanical postprocessing will be described by the entanglement distillation rate $r_\mathrm{ent} = \frac{N}{M} \in [0,1]$. Similar to that, we use $r_\mathrm{sift} = \frac{n}{N} \in [0,1)$ to describe the sifting rate of Step 3. Lastly, the efficiency of classical postprocessing will be denoted by the key distillation rate $r_\mathrm{key} = \frac{\ell}{n}$ which we will model using information-theoretic concepts, that also allow for negative values of $\ell$ resulting in $r_\mathrm{key}\leq 1$ and $r_\mathrm{key} \in \mathbb{R}$. A negative value for $\ell$ and therefore for $r_\mathrm{key}$ in this context can be interpreted such that no secure key can be generated at all. Finally, taking all this into account, we rewrite our original performance metric of (\ref{eq:r}) into 
\begin{equation}
    r = r_\mathrm{ent} \cdot r_\mathrm{sift} \cdot r_\mathrm{key},\ r \in \mathbb{R} \ \mathrm{with}\ r < 1. 
\end{equation}
In the following subsections, we present details on how we model each of these four steps of the \ac{QKD} protocol and what open questions remain with regard to implementing them to achieve the highest key rate $r$ as possible.

\subsection{Entangled Link Generation}\label{CH:Entangled_Link_Generation}
The first step before the quantum states can be processed is the distribution of such states between the end nodes A and B. To do so, we begin with an arbitrary party that locally creates a bipartite entangled quantum system. We assume, like many others within the literature, the Bell pair 
\begin{equation}\label{eq:Phi^+}
    \ket{\Phi^+} = \frac{1}{\sqrt{2}}(\ket{00}+\ket{11})\ \in \mathcal{H}_{AB},
\end{equation}
where $\mathcal{H}_{AB}$ is the joint Hilbert space of A and B. In the next step, this Bell state needs to be transmitted over a quantum channel to the receivers A and B. In reality this distribution unavoidably takes place over a noisy quantum channel and also, due to the fact that the quantum channel is public, enables Eve to interact with the quantum system, i.e., to attack the transmission. We denote the received noisy, potentially attacked and most likely mixed state at A and B as $\varrho_{AB}^0$. The legitimate receivers A and B have no way to determine which of the perturbations of the received quantum system $\varrho_{AB}^0$ (compared to the perfect Bell state $\ket{\Phi^+}$) can be accounted to interactions with the environment in the form of noise and which to the potential attack of E. Therefore, we assume the worst case, i.e., that all deviations are caused by E interacting with the system. This means that after the distribution, there exists a pure quantum state $\varrho_{ABE}$ that is a purification of $\varrho_{AB}^0$ such that $\tr_E[{\varrho_{ABE}}] = \varrho_{AB}^0$.

We assume the initially distributed state $\varrho_{AB}^0$ between A and B to be in the form of a Werner state \cite{Werner1989}
\begin{equation}\label{eq:Werner_State}
    \varrho_{AB}^0 = \frac{4F-1}{3}\ket{\Phi^+}\bra{\Phi^+} + \frac{1-F}{3} I_4,
\end{equation}
which can be fully described by its fidelity $F$. We always refer to the fidelity with respect to the desired Bell state $\ket{\Phi^+}$, i.e., for any density operator $\varrho \in \mathcal{D}(\mathcal{H}_{AB})$ its fidelity is given by 
\begin{equation}\label{eq:F}
    F(\varrho) = \bra{\Phi^+}\varrho\ket{\Phi^+}.
\end{equation}
In the following, we denote the space of all Werner states acting upon the Hilbert space $\mathcal{H}_{AB}$ as $\mathcal{D}_\mathrm{W}(\mathcal{H}_{AB}) \subset\mathcal{D}(\mathcal{H}_{AB})$. This modelling is accurate if the quantum channel (including the potential attack by E) is identical to that of a depolarizing channel (which is often assumed within the literature) because in this case the output is a Werner state. However, this modelling is also justified for the general case due to Werner twirling \cite{Bennett1996b}, which enables the transformation of any arbitrary quantum state to a Werner state as given in (\ref{eq:Werner_State}), without changing the state's fidelity. If we now assume that the goal of the two legitimate users A and B is opposed to that of E, then we argue that at least one of the parties applies Werner twirling, resulting in a Werner state as given within (\ref{eq:Werner_State}).

\subsection{Preprocessing via Entanglement Distillation}\label{CH:r_ent}
Entanglement distillation protocols are a class of protocols that only rely on \ac{LOCC} to transform a set of $M$ quantum states, of a lower quality, into a smaller set of $N < M$ states with a higher quality \cite{Duer2007}. Usually the fidelity $F$ of (\ref{eq:F}) is used to describe the quality, but not exclusively \cite{Kokorsch2025}. Within the literature, one can find several entanglement distillation protocols, e.g., \cite{Bennett1996}, \cite{Deutsch1996}. Furthermore, these protocols can be executed in different configurations \cite{Duer2007}, e.g., recurrence schemes, pumping schemes, or combinations of both usually called nested pumping schemes. Within the following work, we confine our investigations on recurrence schemes. This is a class of iterative protocols, that can be applied for a total of $k \in \mathbb N_0$ iterations, where each of the $d \in \{1,2,\dots k\}$ iterations consumes two identical input states $\varrho_{AB}^{d-1}$, to produce an output state $\varrho_{AB}^{d}$ with a success probability $p_\mathrm{ent}\left(\varrho_{AB}^{d-1}\right) \in [0,1]$. Based on this, the resulting average entanglement distillation rate is given by \cite{Kokorsch2024}
\begin{equation}\label{eq:r_ent(k)}
    r_\mathrm{ent}\left(\varrho_{AB}^0, k\right)= \prod_{d=1}^k \frac{p_\mathrm{ent}\left(\varrho_{AB}^{d-1}\right)}{2}.
\end{equation}
If no entanglement distillation is performed, i.e., $k=0$, no resources are used for this step and the resulting entanglement distillation rate is therefore defined to be $r_\mathrm{ent}(\varrho_{AB}^0,0) = 1\ \forall \varrho_{AB}^0 \in \mathcal{D}(\mathcal{H}_{AB})$. The specific values for the success probability $p_\mathrm{ent}\left(\varrho_{AB}^{d-1}\right)$ and the relationship between the input states $\varrho_{AB}^{d-1}$ and the output states $\varrho_{AB}^{d}$ depend both on the chosen entanglement distillation protocol. For the common entanglement distillation protocols proposed in \cite{Bennett1996} and \cite{Deutsch1996}, the post distillation states $\varrho_{AB}^k$ are in the form of Bell diagonal states
\begin{equation}\label{eq:Bell_Diagonal}
    \begin{split}
    \varrho_{AB}^k = F&\ket{\Phi^+}\bra{\Phi^+} + \delta \ket{\Psi^-}\bra{\Psi^-} \\ 
    + \epsilon &\ket{\Psi^+}\bra{\Psi^+} + \tau \ket{\Phi^-}\bra{\Phi^-}.
    \end{split}
\end{equation}
Here $F$ is again the fidelity of the state, like within (\ref{eq:Werner_State}), and the coefficients $\delta,\epsilon,\tau \in [0,1]$ have the property that $F+\delta+\epsilon+\tau=1$. The output fidelity is monotonically increasing, i.e., $F(\varrho_{AB}^{d}) > F(\varrho_{AB}^{d-1})$, for as long as the input state $\varrho_{AB}^{d-1}$ of an iteration fulfills the criteria of being a distillable state \cite{Duer2007}. We denote the space of all density operators acting on the Hilbert space $\mathcal{H}_{AB}$ within the form of a Bell diagonal state as defined within (\ref{eq:Bell_Diagonal}) as $\mathcal{D}_\mathrm{B}(\mathcal{H}_{AB}) \subset \mathcal{D}(\mathcal{H}_{AB})$. In \cite{Kokorsch2025}, it was shown that, in addition to an increased fidelity, the structural changes in the state caused by entanglement distillation also have a positive influence on the quality of the state. But, since any bipartite state can be transformed into a Werner state via Werner twirling \cite{Bennett1996b}, we restrict our initial analysis to states of the form given in (\ref{eq:Werner_State}). This allows for a more general treatment, as any state can be mapped to this form prior to processing. Moreover, Werner states are fully characterized by a single parameter, their fidelity, which makes them a particularly convenient starting point for our investigation. In Section \ref{CH:Post_Distillation}, we then further investigate how the more general Bell diagonal form affects our analysis. 

Within the literature, it is well investigated how to optimally perform entanglement distillation with the asymptotic goal of generating pure Bell states \cite{Duer2007}, for which the entanglement distillation rate $r_\mathrm{ent}$ of (\ref{eq:r_ent(k)}) goes to $0$. As this is not desirable within the context of practical applications, one common strategy within the literature is to rely on performing entanglement distillation until a certain threshold fidelity is reached, e.g., \cite{Chen2024}, but what this threshold fidelity is supposed to be is often not further discussed. Based on this, we state our first research question \emph{RQ1: How much preprocessing, i.e., how many iterations $k$ of entanglement distillation are optimal with respect to maximizing the resulting final key rate?} But before we can investigate the optimal amount of distillation iterations $k_\mathrm{opt}$ in Section \ref{CH:Optimal_k}, we first need to quantify how the quality of the entangled pairs influences the processing of them. 

\subsection{Entangled State Processing}\label{CH:r_sift}

After entanglement distillation is finished, there remain $N$ copies of the shared entangled state $\varrho_{AB}^k$ between A and B, which are further processed by projection valued measurements. To do so, A measures the observable $\hat{A}(\mathcal{B}_A) = \alpha \ket{\alpha}\bra{\alpha} +\alpha^\perp \ket{\alpha^\perp}\bra{\alpha^\perp}$ on her subsystem, and B the observable $\hat{B}(\mathcal{B}_B) = \beta \ket{\beta}\bra{\beta} +\beta^\perp \ket{\beta^\perp}\bra{\beta^\perp}$ on his subsystem. As we will show in the following subsection, the achievable secure key generation rate only depends on the orthonormal measurement bases $\mathcal{B}_A = \{\ket{\alpha}, \ket{\alpha^\perp}\}$, with $\ket{\alpha},\ket{\alpha^\perp}\in \mathcal{H}_A$ and $\mathcal{B}_B = \{\ket{\beta}, \ket{\beta^\perp}\}$, with $\ket{\beta},\ket{\beta^\perp}\in \mathcal{H}_B$. We therefore do not have to map the measurement results to certain values, and only use them as references to the respective post measurement states. Thus, we define our measurements only based on their associated measurement bases for the remainder of this work. 

The goal of any entanglement based \ac{QKD} protocol is to perform these measurements in such a way that a secure key can be generated between A and B. For this, two generally different kinds of measurements need to be performed. While some measurements are needed to verify the security of the protocol, which is usually done by performing a Bell test based on the \ac{CHSH} inequality \cite{Clauser1969}, other measurements have the goal of extracting classical correlations out of the quantum states that can be used to distill the secure key. To achieve this, A and B each measure randomly in one out of three different measurement bases, which we call $\mathcal{B}_{Ai}$ with $i\in\{1,2,3\}$ for A and respectively $\mathcal{B}_{Bj}$ with $j\in\{1,2,3\}$ for B. 

By doing so, there are nine different combinations in which A and B can collectively measure their shared state $\varrho_{AB}^k$. We denote the set of all these combinations as $\{(\mathcal{B}_{Ai},\mathcal{B}_{Bj})\}_{(i,j) \in \mathcal{S}_\mathrm{meas}}$ with $\mathcal{S}_\mathrm{meas} = \{(i,j)\ | i,j\in\{1,2,3\}\}$. The idea of the protocol is to use the subset $\{(\mathcal{B}_{Ai},\mathcal{B}_{Bj})\}_{(i,j) \in \mathcal{S}_\mathrm{CHSH}}$ with $\mathcal{S}_\mathrm{CHSH} = \{(i,j)\ |i,j\in\{1,2\}\}$ to perform a Bell test by estimating the \ac{CHSH} value, and the remaining five measurements, denoted by the set $\{(\mathcal{B}_{Ai},\mathcal{B}_{Bj})\}_{(i,j) \in \mathcal{S}_\mathrm{proc}}$ with $\mathcal{S}_\mathrm{proc} = \mathcal{S}_\mathrm{meas} \setminus \mathcal{S}_\mathrm{CHSH}$, i.e., all those combinations where at least either A chooses the basis $\mathcal{B}_{A3}$ or B chooses the basis $\mathcal{B}_{B3}$, are used to extract correlations out of the shared quantum state $\varrho_{AB}^k$ that are required for the distillation of the secure key. 

We define the probability that X chooses the measurement basis $\mathcal{B}_{X} = \mathcal{B}_{Xi}$ as $P_{Xi}$ with $X\in\{A,B\}$ and $i\in \{1,2,3\}$. Based on this, we define the joint probability that A chooses $\mathcal{B}_{A} = \mathcal{B}_{Ai}$ and that B chooses $\mathcal{B}_{B} = \mathcal{B}_{Bj}$ as $P_{Ai,Bj} = P_{Ai} P_{Bj}$ with $(i,j)\in \mathcal{S}_\mathrm{meas}$. We assume that the probability for a measurement that can be used for the \ac{CHSH} estimation shall be equal to $\eta \in\ (0,0.25)$ for all the four existing cases and will be dictated by the implementation of the protocol, i.e., $P_{Ai,Bj} = \eta \ \forall(i,j)\in\mathcal{S}_\mathrm{CHSH}$. Even if $\eta$ can be assumed to be fixed, based on the specific implementation, there is still one degree of freedom within choosing how to set the probabilities for the different measurement bases, that we will call $\gamma \in [2\eta,0.5]$. Table \ref{tab:E91_prob} shows the resulting probabilities for the choices of measurement bases for A and B as well as the resulting joint probabilities for the resulting joint measurements. 

\begin{table}[t]
\centering
\caption{Distribution of Measurement Bases for the Entangled Link Processing}
\label{tab:E91_prob}
\renewcommand{\arraystretch}{2}
\setlength{\tabcolsep}{4pt}
\begin{tabular}{c!{\vrule width 1.5pt}c c c}
\diagbox{$\mathcal{B}_B$ for Bob}{$\mathcal{B}_A$ for Alice} & $P_{A1} =\frac{\eta}{\gamma}$ & $P_{A2} =\frac{\eta}{\gamma}$ & $P_{A3} =1-2\frac{\eta}{\gamma}$ \\
\noalign{\hrule height 1.5pt}
$P_{B1} =\gamma$ & $\eta$ & $\eta$ & $\gamma-2\eta$\\ 
$P_{B2} =\gamma$ & $\eta$ & $\eta$ & $\gamma-2\eta$ \\
$P_{B3} =1-2\gamma$ & $\frac{\eta}{\gamma}-2\eta$ & $\frac{\eta}{\gamma}-2\eta$ & $\frac{(2\gamma -1)(2\eta-\gamma)}{\gamma}$ \\
\end{tabular}
\end{table}

This means that there will be five different combinations of chosen measurement bases by A and B whose measurement results can potentially be used to distill the secure key. As it is not guaranteed that all of these five cases will be equally suitable for extracting the correlations necessary for generating a secure key, they need to be investigated and treated separately. The sifting rate for each of these individual combinations is equal to the probability that the corresponding pair of measurement bases get chosen by A and B. We combine the total influence of the key sifting and key distillation process into the processing rate $r_\mathrm{proc} = \frac{\ell}{N} \in [0,1)$. As the key distillation rate behaves linear with respect to the secure key generation rate (see Section \ref{CH:System_Model}) this processing rate is given by the sum of the sifting rate times the key distillation rate of each of these five combinations. This also follows directly from the fact, that we can independently distill keys for each of the useful combinations of chosen measurement bases by A and B, and then add those individual keys up to one final key. Based on this, the processing rate can be given by 
\begin{multline}\label{eq:r_proc}
    r_\mathrm{proc}(\varrho_{AB}^k,\mathcal{B}_{A3},\mathcal{B}_{B3},\eta,\gamma) = \\
    \sum_{(i,j)\in\mathcal{S}_\mathrm{proc}} P_{Ai,Bj} \cdot \max\{0,r_\mathrm{key}(\varrho_{AB}^k,\mathcal{B}_{Ai},\mathcal{B}_{Bj})\}.
\end{multline}
Here, the key distillation rate $r_\mathrm{key}(\varrho_{AB}^k,\mathcal{B}_{Ai},\mathcal{B}_{Bj})$ describes the rate of the postprocessing being performed on the measurement results obtained when measuring $\varrho_{AB}^k$ within the bases $\mathcal{B}_{Ai}$ and $\mathcal{B}_{Bj}$, and will be described in more detail within Subsection \ref{CH:r_key}, see (\ref{eq:Devetak_Winter}). As already mentioned within Section \ref{CH:System_Model}, the key distillation rate is based on quantum information-theoretic models, allowing the rate to become negative, indicating that no secure key can be generated. In these cases, it is more efficient to simply discard these measurements, which is modelled within (\ref{eq:r_proc}) by cutting off the key rate at $0$. While the choice of the measurement bases used for estimating the \ac{CHSH} value, i.e., $\mathcal{B}_{Ai}$ and $\mathcal{B}_{Bj}$ for $(i,j) \in \mathcal{S}_\mathrm{CHSH}$, and the probability $\eta$ are dictated by the \ac{CHSH} inequality and its necessity for verifying the security of the protocol, the bases $\mathcal{B}_{A3}, \mathcal{B}_{B3}$, and $\gamma$ can be chosen freely. Thus, different processing strategies are possible and depend on the collective choice of $\mathcal{B}_{A3}, \mathcal{B}_{B3}$, and $\gamma$. Based on this, our second research question is \emph{RQ2: What processing strategy utilizes the given entangled links $\varrho_{AB}^k$ best, with respect to maximizing the resulting final key rate?} But before we can investigate different processing strategies, we need to further understand how the choice of the measurement bases influences the key distillation rate. To do so, we quantify the rate $r_\mathrm{key}(\varrho_{AB}^k,\mathcal{B}_{Ai},\mathcal{B}_{Bj})$, with which we can distill a secure key out of the available measurement results as shown in the next subsection.

\subsection{Postprocessing Based Key Distillation Capacity}\label{CH:r_key}

After $N$ states of $\varrho_{AB}^k$ are measured, there are $(1-4\eta)N$ measurement results that are not required for the estimation of the \ac{CHSH} value and, thus, can be used to distill the secure key. Since in any realistic system, the state $\varrho_{AB}^k \neq \ket{\Phi^+}\bra{\Phi^+}$, classical postprocessing is necessary. This process includes information reconciliation \cite{Brassard1994}, a protocol that fixes discrepancies in the bit strings obtained by A and B without revealing the whole strings and privacy amplification \cite{Bennett1995} to guarantee that E has no information of the resulting secure key. The efficiency of these processes depends on the measurement distributions, which are given by the measured quantum state $\varrho_{AB}^k$ as well as the chosen measurement bases of Alice $\mathcal{B}_A$ and Bob $\mathcal{B}_B$. Within the context of the key distillation rate $r_\mathrm{key}$, we will call the measured state $\varrho_{AB}^k$ simply $\varrho_{AB}$ and omit the superscript denoting the amount of distillation iterations $k$ that have been performed as described in Subsection \ref{CH:r_ent}. To quantify this key distillation rate $r_\mathrm{key}$, we use the asymptotic one-way secret key capacity bound introduced by Devetak and Winter \cite{Devetak2005}
\begin{equation}\label{eq:Devetak_Winter}
    r_\mathrm{key}(\varrho_{ABE}, \mathcal{B}_A, \mathcal{B}_B) = I(A:B) - \chi(A:E),
\end{equation}
which is an upper bound to the secret key capacity within finite scenarios \cite{Ettinger2025}. 
The classical mutual information $I(A:B)$ between the classical bit strings, after A and B perform a local projection-valued measurement on their respective subsystem, within the basis $\mathcal{B}_A$ and $\mathcal{B}_B$, respectively, can be given by 
\begin{equation}
    I(A:B) = H(A) + H(B) - H(A,B).
\end{equation}
Here $H(\square)$ stands for the classical Shannon entropy and is defined as
\begin{equation}\label{eq:H()}
    H(X) = - \sum_{x \in \sigma[\hat{X}(\mathcal{B}_X)]} p_X(x) \ld p_X(x),
\end{equation}
with the measurement probability
\begin{equation}
    p_X(x) = \bra{x}\varrho_X \ket{x},
\end{equation} 
and $X$, $x$ representing either A, $a$ or B, $b$ with their respective observables $\hat{A}(\mathcal{B}_A)$ and $\hat{B}(\mathcal{B}_B)$. As an extension of (\ref{eq:H()}), the classical joint Shannon entropy is given by
\begin{equation}\label{eq:H(A,B)}
    H(A,B) = - \sum_{\substack{a \in \sigma[\hat{A}(\mathcal{B}_A)] \\ b\in \sigma[\hat{B}(\mathcal{B}_B)]}} p_{A,B}(a,b) \ld p_{A,B}(a,b),
\end{equation}
with 
\begin{equation}\label{eq:p(a,b)}
    p_{A,B}(a, b) = (\bra{a} \otimes \bra{b}) \varrho_{AB} (\ket{a} \otimes \ket{b}). 
\end{equation}
Thus, the classical mutual information $I(A:B)$ only depends on the shared quantum state $\varrho_{AB}$ and the measurement bases $\mathcal{B}_A$ and $\mathcal{B}_B$ chosen by A and B. On the contrary, the Holevo bound 
\begin{equation}\label{eq:Holevo}
    \chi(A:E) = S(\varrho_E)-\sum_{a\in \sigma[\hat{A}(\mathcal{B}_A)]}  p_A(a)S(\varrho_E^a),
\end{equation}
is based on the von Neumann entropy 
\begin{equation}\label{eq:S()}
   S(\varrho) = - \sum_{\lambda\in \sigma[\varrho]} \lambda \ld \lambda,
\end{equation}
and involves the challenge, that it is dependent on the quantum state within Eves possession, i.e., $\varrho_E$ and $\varrho_E^{a}$. Here $\varrho_E^{a}$ is the state in Eve's possession after Alice has measured the result $a$, belonging to the post measurement state $\ket{a} \in \mathcal{B}_A$, on her subsystem. This modelling of the key distillation rate as given by (\ref{eq:Devetak_Winter}), therefore justifies the statement that we made in the previous subsection, i.e., that only the measurement bases are of importance for the key distillation rate and not their explicit mapping to the used observables.

Based on the ideas of \cite{Ferenczi2012} the work of \cite{Ettinger2025} has shown that the influence of the quantum state in E's possession can be fully described via the shared quantum state $\varrho_{AB}$ between A and B if we assume that $\varrho_{ABE}$ is a pure state, as we do (see Subsection \ref{CH:Entangled_Link_Generation}). Since the marginal entropies of pure states are equal, it holds that $S(\varrho_E) = S(\varrho_{AB})$. The measurement of Alice on her subsystem changes the overall system into a \ac{cqq} state $\varrho_{aBE} = \ket{a}\bra{a} \otimes \varrho_{BE}^a$, with probability $p_A(a)$. Since we assume the measured state $\varrho_{ABE}$ to be a pure state, the post measurement state $\varrho_{aBE}$ is also a pure state. From this, it directly follows that the subsystem $\varrho_{BE}^a = \tr_A[\varrho_{aBE}]$ is also a pure state. Hence, we can use the same argumentation as before to state that $S(\varrho_E^a) = S(\varrho_B^a)$. This finally allows us to rewrite the Holevo information of (\ref{eq:Holevo}) into 
\begin{equation}
\label{eq:Holevo_AB}
    \chi(A:E) = S(\varrho_{AB}) - \sum_{a \in \sigma[\hat{A}(\mathcal{B}_A)]} p_A(a)S(\varrho_B^a),
\end{equation}
which allows us to represent the Holevo information between A and E only dependent on the known system of Alice and Bob. Furthermore, based on this, we can rewrite the key distillation rate of (\ref{eq:Devetak_Winter}) into 
\begin{equation}\label{eq:Key_Rate}
    \begin{split}
        r_\mathrm{key}(\varrho_{AB}, \mathcal{B}_A, \mathcal{B}_B) =& H(A)+H(B)-H(A,B)-S(\varrho_{AB}) \\
        & + \sum_{a \in \sigma[\hat{A}(\mathcal{B}_A)]} p_A (a)S(\varrho_B^a).
    \end{split}
\end{equation}

As described within the previous Subsection \ref{CH:r_sift}, one critical step of the processing strategy is choosing the measurement bases of A and B. Therefore, our third and final research question is \emph{RQ3: What combination of measurement bases $\mathcal{B}_A$ and $\mathcal{B}_B$, chosen by A and B, is optimal with respect to the key distillation rate when processing a bipartite entangled state $\varrho_{AB}$ and what is this optimal rate?} As this is the first research question whose answer is not dependent on one of the other research questions we will start with investigating \emph{RQ3} in the following Section. 

\section{Optimal Measurement Bases $\mathcal{B}_A$ and $\mathcal{B}_B$}\label{CH:Optimal_key_rate}
The answer to \emph{RQ3} is given by finding the optimal key distillation rate 
\begin{equation}\label{eq:r_key_opt}
    r_\mathrm{key}^\prime(\varrho_{AB}) = \max_{\mathcal{B}_A, \mathcal{B}_B} r_\mathrm{key}(\varrho_{AB}, \mathcal{B}_A, \mathcal{B}_B). 
\end{equation}
Before we solve this optimization problem, we further simplify the expression (\ref{eq:Key_Rate}) for the case that $\varrho_{AB} \in \mathcal{D}_\mathrm{B}({\cal H}_{AB})$, which is the case for any potentially processed state, as stated in Subsection \ref{CH:r_ent}. Firstly we use the fact that independent of the chosen measurement bases $\mathcal{B}_A$ and $\mathcal{B}_B$, as well as the measurement outcomes $a \in \sigma[\hat{A}(\mathcal{B}_A)]$ and $b \in \sigma[\hat{B}(\mathcal{B}_B)]$, it holds for any $\varrho_{AB} \in \mathcal{D}_\mathrm{B}(\mathcal{H}_{AB})$ that
\begin{equation}\label{eq:p_A}
    p_A (a) = p_B (b) = \frac{1}{2} ,
\end{equation}
which is due to the fact that both of the reduced states $\varrho_A = \tr_B(\varrho_{AB})$ and $\varrho_B = \tr_A(\varrho_{AB})$ are equal to 
\begin{equation}\label{eq:rho_A}
\varrho_A = \varrho_B = \Pi_2.     
\end{equation}
From this, it directly follows that for any $\varrho_{AB} \in \mathcal{D}_\mathrm{B}(\mathcal{H}_{AB})$, and independent of $\mathcal{B}_A$ and $\mathcal{B}_B$, it holds that
\begin{equation}
    H(A) = H(B) = 1.    
\end{equation}
Next, we use the following relationship between the resulting post measurement states when A measured her subsystem. 
\begin{proposition}
    The von Neumann entropy of the conditional post measurement state $\varrho_B^a$, that results from A measuring the state $\varrho_{AB} \in \mathcal{D}_\mathrm{B}(\mathcal{H}_{AB})$ within the orthonormal basis $\mathcal{B}_A = \{\ket{\alpha},\ket{\alpha^\perp}\}$, depends on the chosen measurement basis $\mathcal{B}_A$ but is independent of the actual measurement result $\alpha$ or $\alpha^\perp$ of A, i.e.,  
    \begin{equation}\label{eq:Prop_1}
       S\left(\varrho_B^\alpha\right) = S\left(\varrho_B^{\alpha^\perp}\right) \ \forall \ket{\alpha},\ket{\alpha^\perp} \in \mathcal{H}_A.
    \end{equation}
    \begin{proof}
        For any bipartite quantum state $\varrho_{AB} \in {\cal D}({\cal H}_{AB})$, it holds that 
        \begin{equation}\label{eq:Partial_Post_Measurement_State}
            \varrho_B = \tr_A[\varrho_{AB}] = \sum_{a\in \sigma[\hat{A}(\mathcal{B}_A)]} p_A(a) \varrho_B^a.
        \end{equation}
        Using this and the statements given by (\ref{eq:p_A}) and (\ref{eq:rho_A}), it follows that for any state $\varrho_{AB} \in \mathcal{D}_\mathrm{B}({\cal H}_{AB})$ 
        \begin{equation}
            \varrho_B = \Pi_2 = \frac{1}{2} \varrho_B^\alpha + \frac{1}{2} \varrho_B^{\alpha^\perp},
        \end{equation}
        independent of the measurement basis $\mathcal{B}_A$ chosen by A. As a result from this, we can directly conclude that $\varrho_B^\alpha + \varrho_B^{\alpha^\perp} = I_2$. Thus, it follows that
        \begin{equation}
            \varrho_B^{\alpha^\perp} = ZX(\varrho_B^{\alpha})^* XZ,
        \end{equation} 
        with $X$ and $Z$ being Pauli operators. For any $\varrho \in {\cal D}({\cal H})$ it holds that $\sigma[\varrho_B^{\alpha}] = \sigma[(\varrho_B^{\alpha})^*]$, which combined with the fact that the von Neumann entropy is invariant under unitary actions, i.e., $S(\varrho) = S(U \varrho U^\dagger)$ for any unitary operator $U$, it follows that $S(\varrho_B^{\alpha}) = S( ZX(\varrho_B^{\alpha})^* XZ)$. Therefore, we proved that the von Neumann entropy of the conditional states of B is independent of the measurements result of A.   
    \end{proof}
\end{proposition}

Using the results in (\ref{eq:p_A}) and (\ref{eq:Prop_1}) we can reformulate the expression of the key distillation rate of Equation (\ref{eq:Key_Rate}) into 
\begin{equation}\label{eq:r_key_reduced}
    r_\text{key}(\varrho_{AB},\mathcal{B}_A,\mathcal{B}_B) = 2-H(A,B)-S(\varrho_{AB}) + S\left(\varrho_B^{\alpha}\right), 
\end{equation}
for any $\varrho_{AB} \in \mathcal{D}_\mathrm{B}({\cal H}_{AB})$. Since the term $S(\varrho_{AB})$ is not dependent on the choice of measurement bases, the optimization problem of (\ref{eq:r_key_opt}) is identical to that of
\begin{equation}\label{eq:max_measurement_bases}
    \max_{\mathcal{B}_A, \mathcal{B}_B} S\left(\varrho_B^{\alpha}\right) -H(A,B).
\end{equation}
And since $\mathcal{D}_\mathrm{W}({\cal H}_{AB}) \subset \mathcal{D}_\mathrm{B}({\cal H}_{AB})$ the above statements also hold for any $\varrho_{AB} \in \mathcal{D}_\mathrm{W}({\cal H}_{AB})$, which we want to investigate for now, as described within Subsection \ref{CH:r_ent} (see Section \ref{CH:Post_Distillation} for the generalization to $\varrho_{AB} \in \mathcal{D}_\mathrm{B}({\cal H}_{AB})$). To solve (\ref{eq:max_measurement_bases}), we use the following convenient representation of the canonical entangled state.
\begin{proposition}\label{prop:bell_state}
    The Bell state $\ket{\Phi^+}$ can be written as 
        \begin{equation}\label{eq:Bell_Decomposition}
            \ket{\Phi^+} = \frac{1}{\sqrt{2}}\left(\ket{\psi}\otimes\ket{\psi^*} + \ket{\psi^\perp}\otimes\ket{\psi^{\perp*}} \right),
    \end{equation}
    for any pair of orthonormal states $\{\ket{\psi}, \ket{\psi^\perp}\}$. 
\begin{proof}
    We define $\ket{\psi} = c_0 \ket{0} + c_1 \ket{1}$, with $c_0, c_1 \in \mathbb{C}$. From the fact that $\{\ket{\psi}, \ket{\psi^\perp}\}$ shall form an orthonormal basis, it follows that $\ket{\psi^\perp} = c_1^* \ket{0} - c_0^* \ket{1}$ and that $|c_0|^2 + |c_1|^2 = 1$. Using this, it follows that independent of $c_0$ and $c_1$, the expression of (\ref{eq:Bell_Decomposition}) yields the well known form of the Bell state $\ket{\Phi^+}$ as given within (\ref{eq:Phi^+}). 
\end{proof}
\end{proposition}

\begin{proposition}\label{prop:r_key_werner}
    For Werner states $\varrho_{AB} \in \mathcal{D}_\mathrm{W}(\mathcal{H}_{AB}) $, the optimal key distillation rate as given by (\ref{eq:r_key_opt}) can be achieved for any arbitrary orthonormal measurement basis $\mathcal{B}_A$ as long as B chooses his measurement basis in accordance to
    \begin{equation}\label{eq:prop_optimal_basis}
        \mathcal{B}_B = \{\ket{\psi^*} | \ket{\psi} \in \mathcal{B}_A\}.
    \end{equation}
    If this optimal choice is made, the key distillation rate equals
    \begin{align}
        r_\mathrm{key}^\prime(\varrho_{AB}) &=  1-S(\varrho_{AB}) \label{eq:r_key_werner_alt} \\ 
        &=  1+F\ld F + (1-F) \ld \frac{1-F}{3}.  \label{eq:r_key_werner}    
    \end{align} 

\begin{proof}
    We begin with showing that for $\varrho_{AB} \in \mathcal{D}_\mathrm{W}(\mathcal{H}_{AB})$, the term $S\left(\varrho_B^{\alpha}\right)$ is not dependent on $\mathcal{B}_A$ at all. To do so, we first compute 
    \begin{equation}
        \varrho_B^{\alpha} = \frac{\tr_A[(M_\alpha \otimes I_2)\varrho_{AB}]}{p_A(\alpha)}
    \end{equation}
    with $M_\alpha = \ket{\alpha}\bra{\alpha}$. For a more streamlined notation we will use the Werner parameter $w = 4\frac{1-F}{3}$ to represent the state $\varrho_{AB}$ of (\ref{eq:Werner_State}). Using the result from (\ref{eq:p_A}) and the representation of $\ket{\Phi^+}$ as given by (\ref{eq:Bell_Decomposition}) with $\ket{\psi} = \ket{\alpha}$, we end up with 
    \begin{equation}
        \varrho_B^{\alpha}  = (1-w)\left(\ket{\alpha} \bra{\alpha}\right)^* + w \Pi_2.
    \end{equation}
    With $\ket{\alpha}$ being a normalized state, i.e., $\braket{\alpha|\alpha}=1$, it holds that $\sigma\left[\varrho_B^{\alpha}\right] = \{\frac{w}{2}, 1-\frac{w}{2}\}$ from which it directly follows that $S\left(\varrho_B^{\alpha}\right)$ is not dependent on $\ket{\alpha}$ and therefore also completely independent of $\mathcal{B}_A$. Thus, we can reformulate our problem of (\ref{eq:max_measurement_bases}) to minimizing the classical joint Shannon entropy $H(A,B)$, which can be computed by (\ref{eq:H(A,B)}). 
    We use again the representation of $\ket{\Phi^+}$ as given in (\ref{eq:Bell_Decomposition}) with $\ket{\psi} = \ket{a}$, and the Werner parameter $w$ to express the general joint measurement probability based on (\ref{eq:p(a,b)}), as
    \begin{equation} \label{eq:p_AB}
        p_{A,B}(a, b) = \frac{1-w}{2}|\braket{b|a^*}|^2+\frac{w}{4}. 
    \end{equation}
    From this it follows that $p_{A,B}(\alpha,\beta)  = p_{A,B}(\alpha^\perp,\beta^\perp)$ and that $p_{A,B}(\alpha,\beta^\perp) = p_{A,B}(\alpha^\perp,\beta) = 0.5 -p_{A,B}(\alpha,\beta)$. In combination with the fact that the classical Shannon entropy $H(A,B)$ of (\ref{eq:H(A,B)}) is a strictly concave function, it follows that the minimum of $H(A,B)$ is achieved for the extreme points of $p_{A,B}(a,b)$. Using the result given by (\ref{eq:p_AB}) and the fact that for normalized quantum states $\ket{a}$ and $\ket{b}$ it holds that $|\braket{b|a^*}| \in  [0,1]$, it follows that the minimal achievable value for $H(A,B)$ can be achieved for any arbitrary chosen measurement basis $\mathcal{B}_A$ if and only if we choose the corresponding measurement basis for B to be $\mathcal{B}_B = \{\ket{\psi^*} | \ket{\psi} \in \mathcal{B}_A\}$, therefore proving the statement of (\ref{eq:prop_optimal_basis}). Using this optimal choice for the measurement bases to compute the statement of (\ref{eq:max_measurement_bases}), one gets $S(\varrho_B^\alpha)-H(A,B) = -1$ which when inserted into (\ref{eq:r_key_opt}) and (\ref{eq:r_key_reduced}), directly results in (\ref{eq:r_key_werner_alt}). Finally, when using (\ref{eq:Werner_State}) and (\ref{eq:S()}), one can derive (\ref{eq:r_key_werner}) based on (\ref{eq:r_key_werner_alt}) and therefore proving Proposition \ref{prop:r_key_werner}.
\end{proof}
\end{proposition}

Thus, Proposition \ref{prop:r_key_werner} answers \emph{RQ3} for the case that the measured state is in the form of a Werner state as given in (\ref{eq:Werner_State}). Further, we can directly see within (\ref{eq:r_key_werner}) that even when choosing the optimal measurement bases and assuming optimal postprocessing, the measured state $\varrho_{AB}$ needs to have a minimum threshold fidelity $F_\mathrm{th} \approx 0.8107$ to be able to achieve a positive key rate. This does not yet answer \emph{RQ1}, but at least gives a lower bound on how much distillation iterations are at least necessary to achieve a positive secure key rate $r$. 

Furthermore, the results given by Proposition \ref{prop:bell_state} and \ref{prop:r_key_werner} show, that the frequently found assumption within the literature—that the measurement of the Bell state $\ket{\Phi^+}$, within the same basis, results in a perfect correlated measurement result—is not generally correct. Rather, perfect correlated measurement results are only obtained, if the condition of (\ref{eq:prop_optimal_basis}) holds. Thus, only if the basis states have no relative phase, both bases need to be equal for perfect correlations.

\section{Processing Strategies}\label{CH:Processing}
Within Subsection \ref{CH:r_sift}, we described the dependence of the processing rate on the distribution of the chosen measurement bases of A and B as well as the parameter $\gamma$. Based on Equation (\ref{eq:r_key_reduced}), and Proposition \ref{prop:r_key_werner}, we can now further investigate the processing rate. Generally speaking, the optimal processing rate is given by
\begin{equation}\label{eq:optimal_r_proc}
    r_\mathrm{proc}^\prime(\varrho_{AB}^k,\eta) = \max_{\mathcal{B}_{A3},\mathcal{B}_{B3},\gamma} r_\mathrm{proc}(\varrho_{AB}^k,\mathcal{B}_{A3},\mathcal{B}_{B3},\eta,\gamma),
\end{equation}
which is not solved generally, but for two specific processing strategies in the following. Before we start with the investigation, we first need to define the measurement bases used for the \ac{CHSH} measurements. For the \ac{CHSH} estimation, the full definition of the observables actually plays a role, but as we are only interested in how the remaining choices of measurements influence the key distillation rate and, thus, the processing rate, we are still fine with only considering the resulting measurement bases. As shown by, e.g., \cite{Aspect2002}, the choice of these observables and therefore, for the resulting measurement bases used to measure the \ac{CHSH} value, is not one fixed quartet, but rather a set of observables that preserves the essential geometric structures used for the computation of the \ac{CHSH} value, where all the quartets of observables differ by a global rotation. Since both strategies which will be investigated in the following, rely on choosing $\mathcal{B}_{A3}$ and $\mathcal{B}_{B3}$ to be equal to one of these bases used for the \ac{CHSH} estimation, their specific choice has no influence. Thus, we orient ourselves on the literature and choose without loss of generality $\hat{A}_1 = Z$, $\hat{A}_2 = X$, $\hat{B}_1=\frac{Z-X}{\sqrt{2}}$ and $\hat{B}_2 = \frac{Z+X}{\sqrt{2}}$, see, e.g., \cite{Ekert1991}. This therefore results in the following measurement bases:
\begin{subequations}\label{eq:E91_Bases}
    \begin{align}
        \mathcal{B}_{A1} &= \{\ket{0}, \ket{1}\}, \\
        \mathcal{B}_{A2} &= \left\{\frac{1}{\sqrt{2}}\big(\ket{0}+\ket{1}\big), \frac{1}{\sqrt{2}}\big(\ket{0}-\ket{1}\big)\right\}, \\
        \mathcal{B}_{B1} &= \left\{\frac{(1+\sqrt{2})\ket{0}+\ket{1}}{\sqrt{2(2+\sqrt{2})}}, \frac{(1-\sqrt{2})\ket{0}+\ket{1}}{\sqrt{2(2-\sqrt{2})}}\right\}, \\
        \mathcal{B}_{B2} &= \left\{\frac{(1-\sqrt{2})\ket0+\ket{1}}{-\sqrt{2(2-\sqrt{2})}}, \frac{(1+\sqrt{2})\ket{0}+\ket{1}}{-\sqrt{2(2+\sqrt{2})}}\right\}.
    \end{align}
\end{subequations}
In the following, we want to investigate and compare two strategies for choosing the remaining bases $\mathcal{B}_{A3}$ and $\mathcal{B}_{B3}$ as well as their measurement probabilities given through $\gamma$. The first strategy is the one that can often be found within the literature, e.g., \cite{Acin2006,Pironio2009}, while the second one is, to the best of our knowledge, a new processing strategy that can be shown to achieve higher rates within certain scenarios. As both strategies rely on choosing $\mathcal{B}_{A3}$ and $\mathcal{B}_{B3}$ to be equal to one of the bases used for the \ac{CHSH} estimation as given within (\ref{eq:E91_Bases}), we first start with investigating the key generation rate of (\ref{eq:r_key_reduced}), for the cases that we choose these combinations of bases. If A and B choose the same basis, we can use Proposition \ref{prop:r_key_werner} to determine the resulting key distillation rate (as none of the states given by (\ref{eq:E91_Bases}) has a relative phase), while for the other cases we can use the following two propositions.

\begin{proposition}\label{prop:waste}
    When performing measurements within the two distinct bases used by either A or B for the \ac{CHSH} estimation, i.e., if A measures within $\mathcal{B}_{A} = \mathcal{B}_{Xi}$ and B within $\mathcal{B}_{B} = \mathcal{B}_{Xj}$ with $X\in\{A,B\}$, $i,j\in\{1,2\}$ and $i \neq j$, no positive key distillation rate can be obtained, independently of the measured state $\varrho_{AB}^k \in \mathcal{D}_\mathrm{W}(\mathcal{H}_{AB})$.
\begin{proof}
    Based on the result of Equation (\ref{eq:r_key_reduced}) and the result of (\ref{eq:p_AB}), we can insert all the possible combinations for the measurement bases to always get the result $p_{A,B}(a,b) = \frac{1}{4}$ for any $\varrho_{AB}^k$ being measured. From this it directly follows that $H(A,B) = 2$, which in turn means that the key distillation rate of (\ref{eq:Devetak_Winter}) can be rewritten into $r_\mathrm{key} = - \chi(A:E)$. And since $ \chi(A:E) \geq 0$, it follows that $r_\mathrm{key} \leq 0$, i.e., that no secure key can be generated when using these measurement bases.
\end{proof}
\end{proposition}

\begin{proposition}\label{prop:r_key_suboptimal}
    When performing measurements within a combination of bases that are also used for the \ac{CHSH} estimation, i.e., if A measures within $\mathcal{B}_{A} = \mathcal{B}_{Xi}$ and B within $\mathcal{B}_{B} = \mathcal{B}_{Yj}$ with $X,Y\in\{A,B\}$, $X \neq Y$ and $i,j\in\{1,2\}$, the resulting key distillation rate will be independent on the choice of $\mathcal{B}_{Xi}$ and $\mathcal{B}_{Yj}$. Further, it will be positive for $F(\varrho_{AB}^k) > F_\mathrm{bnd}$, but will always be less than half of what can be achieved when choosing the optimal measurement bases, i.e., 
    \begin{equation}\label{eq:r_key>r_key^prime}
        \frac{1}{2} r_\mathrm{key}^\prime(\varrho_{AB}^k) > r_\mathrm{key}(\varrho_{AB}^k,\mathcal{B}_{Xi},\mathcal{B}_{Yj}).
    \end{equation}
\begin{proof}
    Based on the results established in the proof of Proposition \ref{prop:r_key_werner}, i.e., the fact that $\varrho_B^\alpha$ is independent of the chosen measurement basis $\mathcal{B}_A$ and the result for the general joint measurement probability of (\ref{eq:p_AB}), the expression for the key distillation rate of Equation (\ref{eq:r_key_reduced}) can be computed for any combination of $\mathcal{B}_{Xi}$ and $\mathcal{B}_{Yj}$, based on the measurement bases defined in (\ref{eq:E91_Bases}), and will always result in the same term, that is strictly monotonically increasing. Further, we compared the resulting term with the result for the optimal key distillation rate of (\ref{eq:r_key_werner_alt}) to derive the bound given by (\ref{eq:r_key>r_key^prime}).
\end{proof}
\end{proposition}
We numerically determined the value for the boundary fidelity to be $F_\mathrm{bnd} \approx 0.895$ by computing the value for which $r_\mathrm{key}(\varrho_{AB}^k,\mathcal{B}_{A1},\mathcal{B}_{B1}) = 0$ holds. Based on these results, we will further investigate two explicit processing strategies in the following subsections that we will call the asymmetric processing strategy and the symmetric processing strategy.

\subsection{Asymmetric Processing Strategy}\label{CH:Classical_Measurement_Strategy}
The original proposal of the \ac{E91} protocol states that A and B shall choose their third measurement basis to be equal to one of the measurement bases of the respective other party, i.e., $\mathcal{B}_{A3} = \mathcal{B}_{Bj}$ and $\mathcal{B}_{B3} = \mathcal{B}_{Ai}$ with $j,i\in \{1,2\}$. It directly can be seen that for the case that A uses $\mathcal{B}_{Ai}$ and B uses $\mathcal{B}_{B3} = \mathcal{B}_{Ai}$ for the measurement, and for the case that A uses $\mathcal{B}_{A3} = \mathcal{B}_{Bj}$ and B uses $\mathcal{B}_{Bj}$, the measurements result in an optimal key distillation rate as shown within Proposition \ref{prop:r_key_werner}. Based on these results it follows that the specific choices of $\mathcal{B}_{A3} = \mathcal{B}_{Bj}$ and $\mathcal{B}_{B3} = \mathcal{B}_{Ai}$ with $j,i\in \{1,2\}$ have no influence, and we therefore assume without loss of generality, from now on, that $\mathcal{B}_{A3} = \mathcal{B}_{B1}$ and $\mathcal{B}_{B3} = \mathcal{B}_{A1}$. Combining this with the results of Proportions \ref{prop:waste} and \ref{prop:r_key_suboptimal}, the asymmetric processing rate can be written as 
\begin{align}\label{eq:r_proc_1}
    r_\mathrm{proc,asym}&(\varrho_{AB}^k,\eta,\gamma) = \nonumber \\
    &\left(P_{A1,B3} + P_{A3,B1}\right) \cdot \max\{0,r_\mathrm{key}^\prime(\varrho_{AB}^k)\}  \\ 
    &+ P_{A3,B3} \cdot \max\{0,r_\mathrm{key}(\varrho_{AB}^k,\mathcal{B}_{B1},\mathcal{B}_{A1})\},\nonumber
\end{align}
for which we can optimize the value of $\gamma$ to get the optimal version of this specific processing strategy.
\begin{proposition}\label{prop:gamma_1}
    The optimal processing rate for the asymmetric processing strategy
    \begin{equation}
        r_\mathrm{proc,asym}^\prime(\varrho_{AB}^k,\eta) = \max_{\gamma \in [2\eta,0.5]} r_\mathrm{proc,asym}(\varrho_{AB}^k,\eta,\gamma)
    \end{equation} 
    is achieved for $\gamma = 2\eta$ or $\gamma = 0.5$ and is given by 
    \begin{equation}\label{eq:r_proc_1_prime}
       r_\mathrm{proc,asym}^\prime(\varrho_{AB}^k,\eta) = \left(\frac{1}{2}-2\eta\right)\max\{0,r_\mathrm{key}^\prime(\varrho_{AB}^k)\} .
    \end{equation}
\begin{proof}
    The property in (\ref{eq:r_key>r_key^prime}) can be extended to 
    \begin{equation}\label{eq:max_r_key>r_key^prime}
        \frac{1}{2}\max\{0,r_\mathrm{key}^\prime(\varrho_{AB}^k)\} \geq \max\{0,r_\mathrm{key}(\varrho_{AB}^k,\mathcal{B}_{B1},\mathcal{B}_{A1})\}, 
    \end{equation}
    which shows that it could only be worth measuring in the basis combination $\mathcal{B}_{A3}$ and $\mathcal{B}_{B3}$ if the corresponding sifting rate, i.e., the probability of $P_{A3,B3}$ to end up in this case, is at least twice that of the case that A and B choose the same measurement basis, which happens with the probability $P_{A1,B3}+P_{A3,B1}$. Since this is the case for no $\gamma$, it is optimal to choose $P_{A3,B3} = 0$, which can be done by either choosing $\gamma = 2\eta$ or $\gamma = 0.5$.
\end{proof}
\end{proposition}
The results of Proposition \ref{prop:gamma_1} are in line with what can be found as the processing strategy used within several works, e.g., \cite{Pironio2009} or \cite{Acin2006}. This strategy exploits the fact that reducing either A or B to only two measurement bases may result in a reduced total number of measurements that can be used for generating a secure key, but by doing so, the measurements that don't need to be discarded will always be within the optimal configuration of measurement bases, resulting in the maximal processing rate for this choice of measurement bases, hence the name asymmetric strategy. Contrary to that, we want to present an alternative in the following subsection that we call the symmetric strategy, which is more efficient under certain conditions.

\subsection{Symmetric Processing Strategy}
While the previous scheme chooses $\gamma$ such that no measurements take place within the measurement basis $\mathcal{B}_{A3}$ and $\mathcal{B}_{B3}$, we will exploit exactly this measurement combination next. We do so by choosing  $\mathcal{B}_{A3} = \mathcal{B}_{B3}$, being equal to one of the four bases $\mathcal{B}_{A1}, \mathcal{B}_{A2}, \mathcal{B}_{B1}$, or $\mathcal{B}_{B2}$ as defined within (\ref{eq:E91_Bases}), and 
\begin{equation}
        \gamma = \argmax_{\gamma \in [2\eta,0.5]} P_{A3,B3}= \sqrt{\eta},
\end{equation}
to maximize the probability of choosing the combination of measurement bases $\mathcal{B}_{A3}$ and $\mathcal{B}_{B3}$. This results in the name giving symmetric distribution for choosing the measurement bases between A and B, i.e., $P_{A1} = P_{A2} = P_{B1} = P_{B2} = \sqrt{\eta}$ and $P_{A3} = P_{B3} = 1-2\sqrt{\eta}$. Without loss of generality, we choose $\mathcal{B}_{A3} = \mathcal{B}_{B3} = \mathcal{B}_{A1}$. Using Propositions \ref{prop:r_key_werner}, \ref{prop:waste}, and \ref{prop:r_key_suboptimal} we end up with the following symmetric processing rate
\begin{equation} \label{eq:proc_2_prime}
        \begin{split}
            r_\mathrm{proc,sym}^\prime(\varrho_{AB}^k, \eta) =
            (1-3\sqrt{\eta}+2\eta) \cdot \max\{0,r_\mathrm{key}^\prime(\varrho_{AB}^k)\} \\
            + (2\sqrt{\eta}-4\eta)\cdot \max\{0,r_\mathrm{key}(\varrho_{AB}^k,\mathcal{B}_{A1},\mathcal{B}_{B1})\}. \quad
        \end{split}
\end{equation}

Using this, we can now compare the two processing strategies for different values of $\eta$ and $F\left(\varrho_{AB}^k\right)$.

\subsection{Comparison of the two Processing Strategies}
\begin{proposition}
    For any $\varrho_{AB}^k\in \mathcal{D}_\mathrm{W}(\mathcal{H}_{AB})$ if $\eta < 0.0625$, the symmetric processing strategy always outperforms the asymmetric processing strategy, while for $F\left(\varrho_{AB}^k\right) > F_\mathrm{bnd}$ this even holds for higher values of $\eta$. 
\end{proposition}
\begin{proof}
    For $F\left(\varrho_{AB}^k\right) \leq F_\mathrm{bnd}$, both processing strategies only rely on measurements within identical and therefore optimal measurement bases as shown by Proposition \ref{prop:r_key_werner}. Thus, the only difference between the two processing strategies in this case is the sifting rate, i.e., the probability with which these optimal measurement bases are chosen by A and B when measuring the state $\varrho_{AB}^k$. For the asymmetric strategy of (\ref{eq:r_proc_1_prime}), the sifting rate is $\frac{1}{2}-2\eta$, whereas for the symmetric strategy of (\ref{eq:proc_2_prime}), it is $1-3\sqrt{\eta}+2\eta$. Consequently, the two sifting rates, and thus the processing rates, are identical at $\eta = 0.0625$, whereas the symmetric processing strategy dominates for any $\eta < 0.0625$.

    For the cases that $F\left(\varrho_{AB}^k\right) > F_\mathrm{bnd}$, the symmetric strategy additionally is able to process suboptimal measurements as presented within Proposition \ref{prop:r_key_suboptimal}, further increasing its processing rate. Thus, with increasing fidelity $F\left(\varrho_{AB}^k\right)$ beyond the boundary fidelity $F_\mathrm{bnd}$, this strategy becomes increasingly more efficient, which in turn increases the value for $\eta$ where the asymmetric strategy is outperforming the symmetric one. 
\end{proof}

\begin{figure}[!t]
\centering
\includegraphics[width=3in]{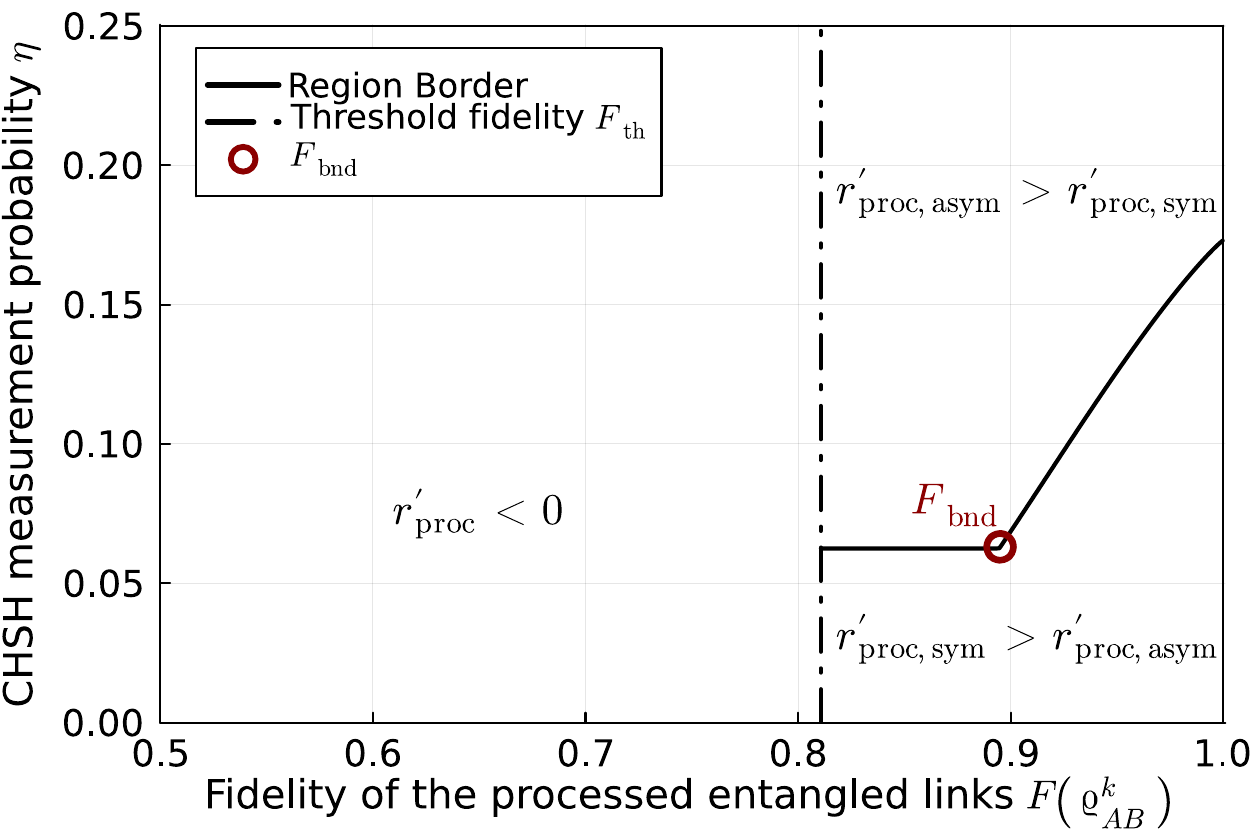}
\caption{Region map comparing the two processing strategies, dependent on the values for $\eta$ and $F(\varrho_{AB}^k)$, and for the assumption that $\varrho_{AB}^k \in \mathcal{D}_\mathrm{W}(\mathcal{H}_{AB})$.} 
\label{fig:Comparison}
\end{figure}

Figure \ref{fig:Comparison} shows the areas, dependent on the values of $\eta$ and $F\left(\varrho_{AB}^k\right)$, where each of the two discussed processing strategies performs best. Additionally, it shows the threshold fidelity $f_\mathrm{th}$ which has been determined in Section \ref{CH:Optimal_key_rate}. With these results, we are not able to fully answer \emph{RQ2} as we did not solve the optimization problem of (\ref{eq:optimal_r_proc}), but we presented a new processing strategy that for certain parameter combinations of $\eta$ and $F\left(\varrho_{AB}^k\right)$ outperforms the common strategy found within the literature. 

\section{Optimal Amount of Entanglement Distillation} \label{CH:Optimal_k}
The final remaining research question is \emph{RQ1}, i.e., finding the optimal amount of entanglement distillation iterations
\begin{equation}\label{eq:k_opt}
    k_\mathrm{opt} =  \argmax_{k\in \mathbb N_0} r_\mathrm{ent}\left(\varrho_{AB}^0,k\right) r_\mathrm{proc}(\varrho_{AB}^k,\mathcal{B}_{A3},\mathcal{B}_{B3},\eta,\gamma),
\end{equation}
for any given processing strategy and entanglement distillation protocol. This results in a trade-off, dependent on $k$, between the quantity and quality of the entangled links. On the one side, with increasing $k$ the entanglement distillation rate will reduce, as each further iteration will require more and more resources resulting in less remaining entangled links. On the other hand, each iteration will increase the quality of the entangled link being processed, resulting in an increasing key distillation rate $r_\mathrm{key}$. For the processing rate, the behavior is not directly clear, as it depends on whether $\eta$ changes depending on $k$ or not. Based on the implementation, $\eta$ could either be defined as a fraction of all $N$ measured states, resulting in $\eta$ being independent of $k$, or as a fixed amount of measurements. In the latter case, $\eta$ would change dependent on $k$, according to $\eta= \frac{\eta_0}{r_\mathrm{ent}\left(\varrho_{AB}^0,k\right)}$ with $\eta_0$ being the desired value of $\eta$ for $k=0$. 

Independent of the influence that entanglement distillation has on the sifting rate and also independent of the exact implementation of the entanglement distillation procedure, and therefore also of the protocols efficiency, the optimal amount of distillation iterations $k_\mathrm{opt}$ can be found numerically, based on the idea in \cite{Ettinger2025}, by confining the search area of $k$. For a more streamlined notation, we do so by omitting the dependency on the variables $\varrho_{AB}^0$, $\mathcal{B}_{A3}$, $\mathcal{B}_{B3}$, $\eta$, and $\gamma$ by expressing the final secure key rate only dependent on $k$, as 
\begin{equation}\label{eq:r(k)}
    r(k) = r_\mathrm{ent}(k) r_\mathrm{proc}(k),
\end{equation}
and the entanglement distillation success probability, cf. Equation (\ref{eq:r_ent(k)}), for the state $\varrho_{AB}^k$ as $p_\mathrm{ent}(k)$. 

\begin{proposition}\label{prop:k_opt}
    After numerically finding the first (potentially local) maximum 
    \begin{equation}\label{eq:k_loc}
        k_\mathrm{loc} = \min\{k\ |\ r(k+1) < r(k)\},
    \end{equation}
    the global maximum $k_\mathrm{opt}$, which is the solution to the problem of (\ref{eq:k_opt}), is either 
    \begin{equation}\label{eq:k_opt_range}
        k_\mathrm{opt} = k_\mathrm{loc} \ \ \text{ or }\ \   k_\mathrm{loc}+1 < k_\mathrm{opt} \leq k_\mathrm{loc} + \kappa,
    \end{equation}
    which means that for $\kappa = 1$ it holds that $k_\mathrm{opt} = k_\mathrm{loc}$. The bounding parameter $\kappa \in \mathbb{N}$ is determined by 
    \begin{equation}\label{eq:kappa}
        \kappa = \min\{\kappa_1,\kappa_2\}-1,
    \end{equation}
    with   
    \begin{equation}\label{eq:kappa_1}
        \kappa_1 = \min\left\{\nu \ |\ \frac{2^\nu r_\mathrm{proc}(k_\mathrm{loc})}{\prod_{x=0}^{\nu-1}p_\mathrm{ent}(k_\mathrm{loc}+x)} > 1-4\eta\right\}
    \end{equation}
    and 
    \begin{equation}\label{eq:kappa_2}
        \kappa_2 = \min\left\{\nu\ |\ r_\mathrm{ent}(k_\mathrm{loc}+ \nu )< \frac{r(k_\mathrm{loc})}{1-4\eta}\right\}.
    \end{equation}

\begin{proof}
    The range in which $k_\mathrm{opt}$ needs to be found, given by (\ref{eq:k_opt_range}), directly follows from the definition of $k_\mathrm{loc}$ in (\ref{eq:k_loc}). 
    
    For the case that $k_\mathrm{opt} \neq k_\mathrm{loc}$, it follows that there exists values for $\nu \in \mathbb{N}$ in which 
    \begin{equation}\label{eq:r_(k+nu)}
        r(k_\mathrm{loc}+ \nu) > r(k_\mathrm{loc}).
    \end{equation}
    Lower bounding the value of $\nu$ for which the inequality (\ref{eq:r_(k+nu)}) can't be true anymore, proves the range of $k_\mathrm{opt}$ of (\ref{eq:k_opt_range}). The first of these bounds, $\kappa_1$ of Equation (\ref{eq:kappa_1}), is proven by first using the definition of the entanglement distillation rate as given in (\ref{eq:r_ent(k)}) to show that 
     \begin{equation}
        r_\mathrm{ent}(k_\mathrm{loc}+\nu) = r_\mathrm{ent}(k_\mathrm{loc})\prod_{x=0}^{\nu-1}\frac{p_\mathrm{ent}(k_\mathrm{loc}+x)}{2}.
    \end{equation}
    Together with (\ref{eq:r(k)}), we rewrite the inequality (\ref{eq:r_(k+nu)}) into
    \begin{equation}
        r_\mathrm{proc}(k_\mathrm{loc}+\nu) \prod_{x=0}^{\nu-1}\frac{p_\mathrm{ent}(k_\mathrm{loc}+x)}{2} > r_\mathrm{proc}(k_\mathrm{loc}),
    \end{equation}
    where we eliminated the term $r_\mathrm{ent}(k_\mathrm{loc})$ from both sides. With the processing rate being limited to $r_\mathrm{proc}(k_\mathrm{loc}+\nu) \leq 1-4\eta$, due to the fact that we have at most a sifting rate of $1-4\eta$ (see Subsection \ref{CH:r_sift}), we can further reformulate this statement as  
    \begin{equation}\label{eq:kappa_1_condition}
        1-4\eta > \frac{2^\nu r_\mathrm{proc}(k_\mathrm{loc})}{\prod_{x=0}^{\nu-1}p_\mathrm{ent}(k_\mathrm{loc}+x)},
    \end{equation}
    being a necessary condition for $\nu$ so that (\ref{eq:r_(k+nu)}) is true. This proves that the lower bound $\kappa_1$ of (\ref{eq:kappa_1}) constrains the solution space of $k_\mathrm{opt}$ as given in the combination of (\ref{eq:k_opt_range}) and (\ref{eq:kappa}).   

    Using again the fact that $r_\mathrm{proc}(k_\mathrm{loc}+\nu) \leq 1-4\eta $, we reformulate the expression of (\ref{eq:r(k)}) as the inequality 
    \begin{equation}
        r(k_\mathrm{loc}+\nu) \leq (1-4\eta)\, r_\mathrm{ent}(k_\mathrm{loc}+\nu).
    \end{equation}
    By combining this inequality with the proposed bound of $\kappa_2$ as given by (\ref{eq:kappa_2}), it follows that any $\nu \geq \kappa_2$ results in a contradiction with the statement of (\ref{eq:r_(k+nu)}). Thus, proving that $\kappa_2$ can also be used to restrict the solution space of $k_\mathrm{opt}$.

\end{proof}
\end{proposition}

As $k_\mathrm{loc}$ exists for any distillable state $\varrho_{AB}^0$, we use Proposition \ref{prop:k_opt} to answer \emph{RQ1} independently of the chosen distillation protocol or processing strategy. We can guarantee to numerically find the optimal amount of distillation iterations $k_\mathrm{opt}$ by beginning with $k=0$ and successively incrementing the value until we first find a local optimum $k_\mathrm{loc}$ and then continue until we reach one of the bounding conditions. Due to the recursive nature of calculating the effects of entanglement distillation (cf. (\ref{eq:r_ent(k)}) or e.g. \cite{Duer2007}), this numerical approach of computing $k_\mathrm{opt}$ has a computational complexity comparable to an analytic solution for finding $k_\mathrm{opt}$, as all distillation iterations $d\leq k_\mathrm{opt}$ would also need to be computed in this case.

\section{Consideration of Post Distillation Shape}\label{CH:Post_Distillation}
As stated within Subsection \ref{CH:r_ent}, post distillation states $\varrho_{AB}^k \ \forall k \geq 1$, are generally in the form of Bell diagonal states as defined in (\ref{eq:Bell_Diagonal}). In \cite{Kokorsch2025}, it has been shown that this is generally a more beneficial type of state than Werner states. Since any arbitrary state can always be transformed into the shape of a Werner state without changing the state's fidelity \cite{Bennett1996b}, the results obtained in Section \ref{CH:Optimal_key_rate} act as a lower bound for states in other shapes. In the following, we want to briefly discuss what changes if we don't make the simplified assumption that the processed state is a Werner state, i.e., $\varrho_{AB}^k \in \mathcal{D}_\mathrm{W}(\mathcal{H}_{AB}) \ \forall k \in \mathbb{N}_0$. Starting with \emph{RQ1}, the results of Section \ref{CH:Optimal_key_rate} are only partially true for Bell diagonal states. We can still reduce the expression of the key distillation rate of (\ref{eq:Devetak_Winter}) to (\ref{eq:r_key_reduced}), but the result for the optimal measurement basis of (\ref{eq:prop_optimal_basis}) in Proposition \ref{prop:r_key_werner} do not apply anymore. This can be simply shown by computing the resulting key distillation rate for the case that both A and B measure either within the basis $\mathcal{B}_{A1}$ or $\mathcal{B}_{A2}$ which will, contrary to the case where $\varrho_{AB}^k \in \mathcal{D}_\mathrm{W}(\mathcal{H}_{AB})$, not result in the same key distillation rate. While we did not solve the optimization problem of (\ref{eq:r_key_opt}) for the case that $\varrho_{AB}^k \in \mathcal{D}_\mathrm{B}(\mathcal{H}_{AB})$, we propose the following partial result, proving parts of the conjecture stated in \cite{Kokorsch2025}. 

\begin{proposition}\label{prop:r_key_bell}
For the case  that $\varrho_{AB}^k \in \mathcal{D}_\mathrm{B}(\mathcal{H}_{AB})$ and if A and B both choose the measurement basis $\mathcal{B}_{A1} = \{\ket{0},\ket{1}\}$, the resulting key distillation rate is identical to the optimum for the case that $\varrho_{AB}^k \in \mathcal{D}_\mathrm{W}(\mathcal{H}_{AB})$ as given in (\ref{eq:r_key_werner_alt}), i.e.,
\begin{equation}\label{eq:r_key}
    r_\mathrm{key} (\varrho) = 1- S(\varrho).
\end{equation}

\begin{proof}
    We start with (\ref{eq:r_key_reduced}) which also holds for this case. Using the same concept as in the proof of Proposition \ref{prop:r_key_werner}, the joint classical entropy $H(A,B)$ can be computed to be 
    \begin{equation}
        H(A,B) = -(F+\tau)\ld\left(\frac{F+\tau}{2}\right) - (\epsilon+\delta)\ld\left(\frac{\epsilon+\delta}{2}\right).
    \end{equation}
    Again, based on the same argumentation as in the proof of Proposition \ref{prop:r_key_werner}, we compute the von Neumann entropy of the reduced post measurement state as 
    \begin{equation}
       S(\varrho_{B}^\alpha) = -(F+\tau)\ld(F+\tau) - (\epsilon+\delta) \ld(\epsilon+\delta).
    \end{equation}
    From this, it follows that 
    \begin{equation}
        S(\varrho_{B}^\alpha) -H(A,B) = -1
    \end{equation}
    which directly proves the statement in (\ref{eq:r_key}).
\end{proof}
\end{proposition}
The fact that this solution is identical to the optimal solution for Werner states as given by (\ref{eq:r_key_werner_alt}), indicates that this is also the optimal solution for the more general case of an arbitrary Bell diagonal state, which partly answers \emph{RQ3}. But even then, it remains unanswered if this is the only combination of measurement bases for which the optimal key distillation rate can be achieved, or if there exist others.
With regard to the processing strategy, i.e., \emph{RQ2}, the situation is similar. The results of Proposition \ref{prop:waste} can be verified to also hold for Bell diagonal states, following the same steps as for the original proof, while the same is not possible for Proposition \ref{prop:r_key_suboptimal}. Thus, no clear statement about the effectiveness of the processing strategies can be made when dealing with Bell diagonal states. Lastly, the question of the optimal amount of distillation iterations stated within \emph{RQ1} can still be fully answered by the results of Proposition \ref{prop:k_opt}, as at no point the assumption that $\varrho_{AB}^k \in \mathcal{D}_\mathrm{W}(\mathcal{H}_{AB})$ was made within its proof.

\section{Conclusion and Outlook}\label{CH:Conclusion}
Within the presented paper, we investigated and discussed the whole process chain necessary for generating a secure key with an entanglement based \ac{QKD} protocol that utilizes the Bell inequality to guarantee its security. To do so, we made the worst-case assumption that the attacker E has full control over the quantum channel and began with confining our investigations to Werner states. Starting from the distribution of the entangled states between A and B, we investigated how to determine the optimal amount of preprocessing in the form of a recurrence based entanglement distillation protocol (\emph{RQ1}). For this we proposed and proved a way of upper bounding the optimal amount of iterations, allowing to quickly find the optimal solution via a numerical search for any given system. Further, we investigated how to process these entangled links, by proposing a new processing strategy and showing that it is better than the common strategy from within the literature for certain scenarios (\emph{RQ2}). Additionally, we answered \emph{RQ3}, i.e, which measurement bases need to be chosen by A and B, to achieve the key distillation capacity introduced by Devetak and Winter, and proved this result. Beyond that, we also briefly discussed how these results change if we don't assume Werner states but model the post-distillation states as arbitrary Bell diagonal states. 

Based on these results, various avenues for future research can be considered. For one, it would be of great interest to actually find and prove what the optimal processing strategy is, i.e., solving the proposed optimization problem of (\ref{eq:optimal_r_proc}). In addition, the influence that the states shape has on the results should be further investigated, involving the optimal measurement bases for, e.g., Bell diagonal states.

\bibliographystyle{IEEEtran}
\bibliography{references}

\end{document}